\documentclass[11pt]{article}

\usepackage[a4paper]{geometry} 

\usepackage[utf8]{inputenc}
\usepackage{graphicx}
\usepackage{amsmath,amssymb,amsthm}
\usepackage[authoryear,round]{natbib}
\usepackage{authblk}

\newtheorem{prop}{Proposition}

\usepackage{xcolor} 
\usepackage{hyperref}

\title{Nonparametric testing of the covariate significance for spatial point patterns under the presence of nuisance covariates}

\author[1]{Ji\v{r}\'{i} Dvo\v{r}\'{a}k}
\author[2]{Tom\'{a}\v{s} Mrkvi\v{c}ka}
\affil[1]{Faculty of Mathematics and Physics, Charles University, Czech Republic}
\affil[2]{Faculty of Economics, University of South Bohemia, Czech Republic}

\begin{document}

\maketitle

\noindent \textbf{Abstract.} Determining the relevant spatial covariates is one of the most important problems in the analysis of point patterns. Parametric methods may lead to incorrect conclusions, especially when the model of interactions between points is wrong. Therefore, we propose a fully nonparametric approach to testing significance of a covariate, taking into account the possible effects of nuisance covariates. Our tests match the nominal significance level, and their powers are comparable with the powers of parametric tests in cases where both the model for intensity function and the model for interactions are correct. When the parametric model for the intensity function is wrong, our tests achieve higher powers. The proposed methods rely on Monte Carlo testing and take advantage of the newly introduced covariate-weighted residual measure. We also define a correlation coefficient between a point process and a covariate and a partial correlation coefficient quantifying the dependence between a point process and a covariate of interest while removing the influence of nuisance covariates.

\vspace{0.25cm}


\section{Introduction}
\subsection{Motivation and overview}
Spatial point patterns are often accompanied by spatial covariates. Determining the relevant covariates that influence the positions of points is certainly one of the most important questions of point pattern analysis. Applications include spatial epidemiology, spatial ecology, exploration geology, seismology, and many other fields. 

In this paper, we mainly focus on this question. Our proposed methods use nonparametric tools. The second question that we are interested in is nonparametric quantification of the spatial dependence between a point process and a covariate, both without and with presence of nuisance covariates. We define a correlation coefficient and a partial correlation coefficient between a point process and a covariate. The second problem has not been studied before, to our knowledge.

The first problem is usually solved by parametric methods \citep{Schoenberg2005,WaagepetersenGuan2009,Kutoyants1998,Coeurjolly2013}, see Section~\ref{subsec:parametric_methods} for details. However, we show in our simulation study that even when the parametric model is selected correctly, these tests of covariate significance may lead to liberality. The parametric methods have even bigger problems when: 1) the parametric model for the intensity function is incorrect, or 2) the form of interactions between points is specified incorrectly. We propose here two tests of covariate significance, a fully nonparametric one which avoids both selecting the intensity function model and the interaction model, and a semiparametric one which does not assume an interaction model but uses the log-linear intensity function model as the one predominantly used in practice. These two proposed tests do not exhibit liberality, and their powers are comparable with the powers of parametric methods in cases with correctly specified models for the intensity function and the interactions. The proposed tests also have a higher power than the parametric ones when either the intensity function model or the interaction model is misspecified.

Since the proposed nonparametric tests do not need to choose a specific model and exhibit better properties than parametric methods, their use should become a standard practice in the analysis of point patterns.

For determining relevant covariates one can also use the lurking variable plots \citep{Baddeley2005a} or appropriate information critera \citep{Choiruddin2021} but these do not provide formal tests. The only nonparametric method studying the dependence of a point process and a covariate without nuisance covariates was introduced in \citet{Dvorak2022}.

Throughout the paper, we assume that the spatial covariates are continuous. The methodology is up to a certain extent also applicable for categorical covariates, as discussed in Section~\ref{sec:CD}.


\subsection{Motivational examples}\label{subsec:motivation}

To illustrate the relevance of the questions posed above, we consider a part of the tropical tree data set from the Barro Colorado Island plot \citep{Condit1998}. We focus on the positions of 3\,604 trees of the Beilschmiedia pendula species in a rectangular $1\,000 \times 500$ metre sampling plot, plotted in the top left panel of Figure~\ref{fig:BCI}. This part of the data set is available in the \texttt{spatstat} package. Below, we call it the BCI data set.

The intensity of point occurrence in the observation window is clearly nonconstant as the trees tend to prefer specific environmental conditions. The variation in the intensity of point occurrence may possibly be explained by the accompanying covariate information. The available covariates include the terrain elevation and gradient (available in the \texttt{spatstat} package) and the soil contents of mineralised nitrogen, phosphorus and potassium \citep{BCIsoil}, see Figure~\ref{fig:BCI}. Maybe all the covariates bring important information and should be used for inference. However, it is equally possible that some of the covariates bring redundant information (as could be expected from the nitrogen and potassium content in this data set, see the bottom left and bottom right panel of Figure~\ref{fig:BCI}) or that some of the covariates, in fact, do not influence the point process. It is important to determine with high degree of confidence which covariates influence the point process and should be included in the further steps of the inference.

In certain cases, a relevant parametric model can be specified based on the available expert knowledge. However, often no such parametric model is available, or we do not want to take a risk of model misspecification. Then nonparametric methods for covariate selection need to be used. 

\begin{figure}[t]
    \centering
    \includegraphics[width=\textwidth]{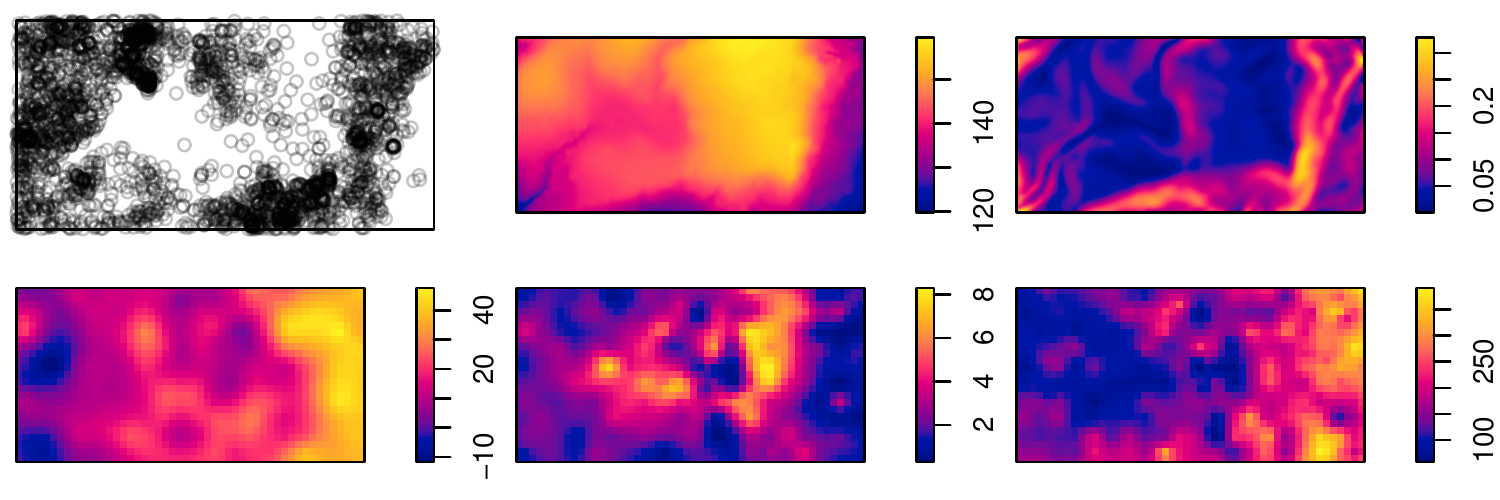}
    \caption{The Barro Colorado Island data set. From left to right, top to bottom: locations of trees, terrain elevation, terrain gradient, the soil contents of nitrogen, phosphorus and potassium.}
    \label{fig:BCI}
\end{figure}

Furthermore, we consider the Castilla-La Mancha forest fire data set, again available in the \texttt{spatstat} package. We study the locations of 689 forest fires that occurred in this region in Spain in 2007, plotted in the left panel of Figure~\ref{fig:CLM}. Below we call it the CLM data set. The size of the region is approximately 400 by 400 kilometers. The intensity of point occurrence is nonconstant and may be influenced by the accompanying covariates (terrain elevation and gradient, see the middle and right panels of Figure~\ref{fig:CLM}). We aim at quantifying the strength of influence of the individual covariates on the point process and comparing it with the BCI data set.

\begin{figure}[t]
    \centering
    \includegraphics[width=\textwidth]{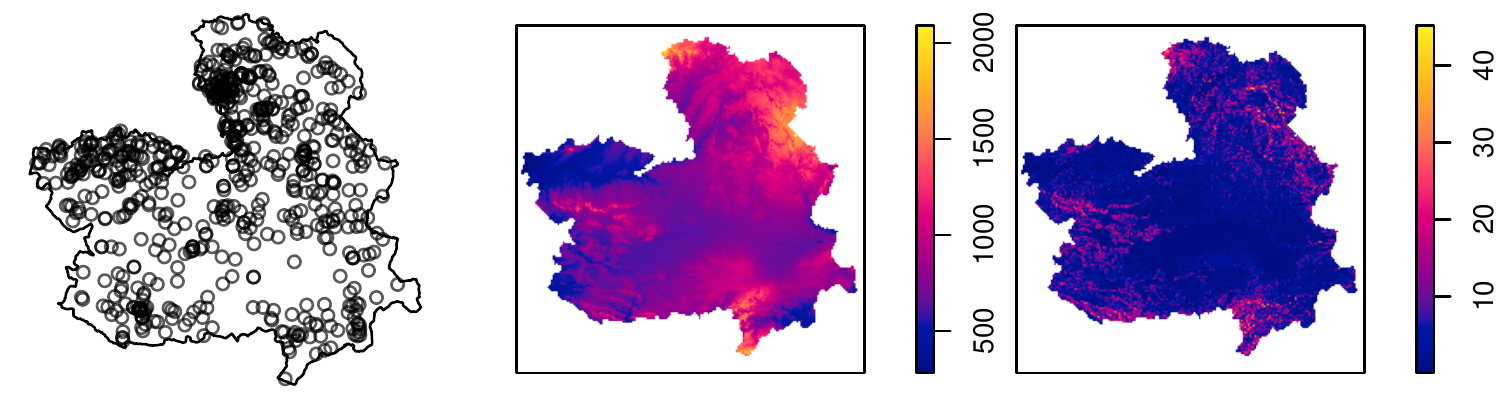}
    \caption{The Castilla-La Mancha data set. From left to right: locations of forest fires, terrain elevation, terrain gradient.}
    \label{fig:CLM}
\end{figure}

\subsection{Outline of the work}
In order to achieve our objectives, we propose to employ the residual analysis \citep{BaddeleyEtal2005} with respect to the model built from the nuisance covariates. The sample (Kendall's) correlation coefficient of the smoothed residual field and the interesting covariate then quantifies their dependence both without and with nuisance covariates. The latter defines the partial correlation.

The testing of covariate significance is proposed to be performed via a new test statistic, the covariate-weighted residual measure, and a Monte Carlo test. The residual analysis can be computed in the parametrical way, which defines our semiparametrical approach, or it can be computed nonparametrically using the nonparametrical estimate of the point pattern intensity \citep{Baddeley2012} and it defines our completely nonparametrical approach. The nonparametric residuals are used  for the first time in this work.  

The replications in the Monte Carlo test are obtained through random shifts  both with torus correction \citep{Lotwick1982} and variance correction \citep{MrkvickaEtAl2020}. The torus correction is a standard method whereas the variance correction was recently defined, and it allows to use nonrectangular windows and it better controls the level of the test than the torus correction. 

The paper is organised as follows. Section~\ref{sec:NB} recalls all the concepts we need to define our procedures. Section~\ref{sec:NM} describes all new methods we are introducing in this work. That is, nonparametric residuals, spatial (partial) correlation coefficient, covariate-weighted residual measure, and tests of covariate significance with nuisance covariate. Section \ref{sec:simulations} contains a simulation study in which the exactness and power of our nonparametrical methods is compared with parametrical methods. Section \ref{sec:application_BCI} contains an example of the usage of our methods for nonparametric selection of relevant covariates. Section \ref{sec:application_CLM}  contains an example of usage of our methods for comparison of dependence strength. Finally, Section~\ref{sec:CD} is left for conclusions and discussion. 

The \texttt{R} codes providing an implementation of the proposed methods are available at \url{https://msekce.karlin.mff.cuni.cz/~dvorak/software.html} and will be available in the planned package \texttt{NTSS} for \texttt{R}.

\section{Notation and background} \label{sec:NB}

Let $X$ be a point process on $\mathbb{R}^2$ with the intensity function $\lambda(u)$. Throughout this paper, we assume that the intensity function of $X$ exists. Let $C_1, C_2, \ldots, C_{m+1}$ be the covariates in $\mathbb{R}^2$. 
Denote by $W \subset \mathbb{R}^2$ a compact observation window with area $|W|$ and $n(X \cap B)$ the number of points of the process $X$ observed in the set $B$. We assume that the values of the covariates are available in all points of $W$, at least on a fine pixel grid. This can be achieved from a finite set of observations, e.g. by kriging techniques.

\subsection{Covariate selection in parametric point process models}\label{subsec:parametric_methods}

The dependence of the intensity function of a point process on the covariates $C_1, \ldots, C_m$ is often modelled parametrically, e.g. using the log-linear model
\begin{align}\label{eq:lambda_loglinear}
    \lambda(u; \beta) = \exp\{\beta_0 + \beta_1 C_1(u) + \ldots + \beta_m C_m(u)\}.
\end{align}
The standard approach to estimating the model parameters $\beta_i$ is to maximize the Poisson likelihood \citep{Schoenberg2005,WaagepetersenGuan2009}. This corresponds to the maximum likelihood approach for Poisson models, while for non-Poisson models, this constitutes a first-order composite likelihood approach. For the log-linear model \eqref{eq:lambda_loglinear} the estimation is implemented in the \texttt{ppm} function from the popular \texttt{spatstat} package \citep{baddeley2015spatialR}.

For Poisson or Gibbs processes, the \texttt{ppm} function also provides confidence intervals for the regression parameters $\beta_i$ and the p-values of the tests of the null hypothesis that $\beta_i = 0$ for a given $i$, based on the asymptotic variance matrix \citep{Kutoyants1998,CoeurjollyRubak2013}.
For cluster processes, the \texttt{kppm} function from the \texttt{spatstat} package provides means of model fitting. The regression parameters $\beta_i$ from \eqref{eq:lambda_loglinear} are again estimated using the \texttt{ppm} function, but the asymptotic variance matrix is determined according to \citet{Waagepetersen2008}, taking into account the attractive interactions between points.

The methods discussed above provide means for formal testing of the hypothesis that $\beta_i = 0$ for a given $i \in \{1, \ldots, m\}$, allowing one to select the set of relevant covariates to be included in the model.

\subsection{Parametric residuals for point processes}

Residuals can be used to check whether the fitted model for the intensity function is appropriate, see \cite{BaddeleyEtal2005} or \citet[Sec.~11.3]{baddeley2015spatialR}. In the following we employ the version of residuals based on the intensity function, as suggested by R. Waagepetersen in the discussion to the paper \cite{BaddeleyEtal2005}, rather than based on the conditional intensity function as discussed in the paper itself. Let $\hat{\beta}$ be the vector of the estimated regression parameters. The \emph{residual measure} is defined as
\begin{align}\label{eq:paramR}
    \mathcal{R}(B) = n(X \cap B) - \int_B \lambda(u;\hat{\beta}) \, \mathrm{d}u,
\end{align}
where $B \subseteq W$ is a Borel set. The \emph{smoothed residual field} is obtained as
\begin{align}\label{eq:paramSRF}
    s(u) = \frac{1}{e(u)} \left[ \sum_{x_i \in X \cap W} k(u-x_i) - \int_W k(u-v) \lambda(v;\hat{\beta}) \, \mathrm{d}v \right],
\end{align}
where $e(u) = \int_W k(u-v) \, \mathrm{d}v$ is the edge-correction factor  
and $k$ is a probability density function in $\mathbb{R}^2$. In fact, the first term in \eqref{eq:paramSRF} gives the nonparametric kernel estimate of the intensity function, the covariates not being taken into account, while the second term gives the smoothed parametric estimate which incorporates the covariates. If the estimated model $\lambda(v;\hat{\beta})$ describes the point process $X$ well, the smoothed residual field $s(u)$ is expected to fluctuate around 0. Its deviations from 0 indicate a disagreement between $\lambda(v;\hat{\beta})$ and the true intensity function in the corresponding parts of the observation window. We remark that the residuals described above are the \emph{raw residuals} of \citet{BaddeleyEtal2005}, where scaled versions of the residuals are also considered.

\subsection{Nonparametric estimation of the intensity function depending on covariates}\label{subsec:nonpar_rho}

As opposed to fitting a parametric model such as \eqref{eq:lambda_loglinear}, the dependence of the intensity function on a set of covariates can be captured nonparametrically. \cite{Baddeley2012} assume that there is an unknown function $\rho: \mathbb{R}^m \rightarrow [0,\infty)$ such that
$\lambda(u) = \rho(C_1(u), \ldots, C_m(u)).$
Assuming absolute continuity of the distribution of the vector of covariates $(C_1(u), \ldots, C_m(u))$ on $\mathbb{R}^m$, the function $\rho$ can be estimated using kernel smoothing in the space of covariate values, see \citet{Baddeley2012} or \citet[Sec.~6.6.3]{baddeley2015spatialR}. This opens up the possibility to define the nonparametric residuals in Section~\ref{subsec:nonparametric_residuals}.

The estimation of $\rho$ is implemented in the \texttt{rhohat} function from the \texttt{spatstat} package for $m=1$ and in the \texttt{rho2hat} function for $m=2$. We note that in these two cases, visualization of $\hat \rho$ is straightforward while it is not as easy for $m > 2$. In our simulation experiments in Section~\ref{sec:simulations} we use the \texttt{spatstat} implementation, while in the analysis of the real data sets with higher number of covariates we use our implementation based on the \texttt{ks} package \citep{ks_package}.

\subsection{Monte Carlo tests}

When the distribution of a test statistic is too complicated to be derived analytically but there is a way of obtaining replications (simulations, permutations, \ldots) of the data under the null hypothesis, it is possible to perform a formal test of the null hypothesis using the Monte Carlo approach \citep{DavisonHinkley1997}. 
This approach relies on the exchangeability of the vector $(T_0, T_1, \ldots, T_N)$, where $T_0$ is the test statistic value computed from the observed data, and $T_1, \ldots, T_N$ are obtained from the replications.

The test is performed by determining how typical or extreme the value $T_0$ is in the whole sample $T_0, T_1, \ldots, T_N$. For univariate test statistics, this means determining the rank of $T_0$, however, using functional test statistics is also possible if a suitable ranking of the functions from the most typical to the most extreme is available, as e.g. in \citet{MyllymakiEtal2017}. Excheangeability (invariance of the distribution with respect to permutations of the components) ensures that the Monte Carlo test matches the required significance level.

\subsection{Random shift permutation strategy}\label{sec:RStests}

Random shifts provide means of nonparametric testing of independence between a pair of spatial objects, such as a pair of random fields \citep{UptonFingleton1985,DaleFortin2002} or a pair of point processes \citet{Lotwick1982}. By randomly shifting one of the objects while keeping the other one fixed, any possible dependence between them is broken. At least one of the spatial objects must be assumed to be stationary. By performing a certain amount of shifts along randomly generated vectors, one obtains replications for performing a Monte Carlo test of independence.

Assume that the spatial objects are denoted by $\Phi$ and $\Psi$ and we observe them in the window $W$. We denote the value of the test statistic computed directly from the observed data by $T_0 = T(\Phi, \Psi; W)$. After producing $N$ random shift vectors $v_1, \ldots, v_N$ we compute the value of the test statistic $T_i$ from $\Phi$ and $\Psi$ shifted by $v_i$, i.e. $T_i = T(\Phi, \Psi + v_i; W)$, $i = 1, \ldots, N$. Clearly, some part of $\Psi$ will be shifted outside of the observation window $W$ and part of $\Psi + v_i$ will not overlap with $\Phi$ anymore. Hence, some form of correction is needed.

\subsubsection{Torus correction}

For a rectangular window $W$, one may identify its opposing edges, creating a toroidal geometry on $W$ \citep{Lotwick1982,UptonFingleton1985}. We denote by $[\Psi + v_i]$ the version of $\Psi$ shifted with respect to the toroidal geometry, as opposed to $\Psi + v_i$ which denotes $\Psi$ shifted with respect to the Euclidean geometry. The replications $T_i$ are then obtained as $T_i = T(\Phi, [\Psi + v_i]; W)$, $i = 1, \ldots, N$.

As a result, all parts of the data are used for computing $T_i$. On the other hand, artificial cracks appear in the correlation structure of the data, as parts of the data originally far away are now ``glued together''. This means that exchangeability is violated, which in turn introduces liberality of the random shift tests \citep{FortinPayette2002,MrkvickaEtAl2020}. However, simulation studies show that when the spatial autocorrelations in the data are not very strong, the tests match the nominal significance level quite closely \citep{MrkvickaEtAl2020,Dvorak2022}. Traditionally, the distribution of the random shift vectors is taken to be the uniform distribution on $W$, but other choices are also possible.

\subsubsection{Variance correction}\label{subsubsec:variance_correction}

To remove the liberality of the torus correction, \citet{MrkvickaEtAl2020} proposed the \emph{variance correction}. It uses shifts respecting the Euclidean geometry and discards those parts of the data that are shifted outside of $W$. No artificial cracks are introduced to the correlation structure of the data, removing the liberality of the random shift tests. Also, irregular observation windows can be considered. On the other hand, different amounts of data are dropped for different shift vectors $v_i$ and for typical choices of the test statistic the variance of $T_i$ varies greatly, making it impossible to perform the Monte Carlo test directly. Therefore, the variance of $T_i$ needs to be standardized before performing the test.

Formally, we denote by $W_i$ the smaller observation window where $\Phi$ and $\Psi + v_i$ overlap, i.e. $W_i = W \cap (W + v_i)$. The value $T_i$ is computed from $\Phi$ and $\Psi + v_i$ restricted to $W_i$, specifically as $T_i = T(\Phi|_{W_i},(\Psi+v_i)|_{W_i};W_i)$. The values $T_0, T_1, \ldots, T_N$ are then standardized to have zero mean and unit variance. This is achieved by subtracting the mean $\overline{T} = \frac{1}{N+1} \sum_{i=0}^N T_i$ and dividing by the square root of the variance:
$
    S_i = \left( T_i - \overline{T} \right) / \sqrt{\text{var} (T_i)}.
$
The standardized values $(S_0, S_1, \ldots, S_N)$ are closer to exchangeability than $(T_0, T_1, \ldots, T_N)$ because their first two moments are the same. The standardized values are used to perform the Monte Carlo test. When a formula describing $\text{var}(T_i)$ as a function of the size of $W_i$ is known, at least asymptotically, it can be directly used in the standardization. If such a formula is not available, \citet{MrkvickaEtAl2020} suggest a kernel regression approach to estimating $\text{var} (T_i)$.

Simulation studies in \citet{MrkvickaEtAl2020,Dvorak2022} show that the random shift tests with variance correction match the nominal significance level even in the case of strong autocorrelation. In those papers, the shift vectors followed the uniform distribution on a disc with radius $R$ centered at the origin. The choice of $R$ is a compromise between two goals: longer shifts are more relevant for breaking the possible dependence between $\Phi$ and $\Psi$ while shorter shifts mean that a larger amount of available data is used to compute $T_i$. Choosing $R$ so that $|W_i| / |W| \geq 1/4$ for all $i$ turned out to provide satisfactory results.

\subsection{Nonparametric testing of dependence between point process and a covariate}\label{subsec:PCtest}

For nonparametric testing of the null hypothesis of independence between a point process $X$ and a covariate $C_1$ the paper \citet{Dvorak2022} suggests to use the random shift test with the test statistic $T = \frac{1}{n(X \cap B)} \sum_{x_i \in X \cap W} C_1(x_i)$, i.e. the mean covariate value observed at the points of the process. This test showed liberality (with torus correction) or slight conservativeness (with variance correction) in the simulation studies in \citet{Dvorak2022}, with both versions having much higher power than the other tests considered there.

\section{New methods} \label{sec:NM}
\subsection{Nonparametric residuals for point processes}\label{subsec:nonparametric_residuals}

As discussed in Section~\ref{subsec:nonpar_rho}, a nonparametric estimate of the intensity function $\hat\lambda(u) = \hat\rho(C_1(u), \ldots, C_m(u))$ can be used to describe its dependence on the set of covariates. Using $\hat\rho$, the nonparametric version of the residual measure \eqref{eq:paramR} can be defined as
\begin{align}\label{eq:nonparR}
    \tilde{\mathcal{R}}(B) = n(X \cap B) - \int_B \hat{\rho}(C_1(u), \ldots, C_m(u)) \, \mathrm{d}u.
\end{align}
The corresponding nonparametric smoothed residual field is then
\begin{align}\label{eq:nonparSRF}
    \tilde{s}(u) = \frac{1}{e(u)} \left[ \sum_{x_i \in X \cap W} k(u-x_i) - \int_W k(u-v) \hat{\rho}(C_1(u), \ldots, C_m(u)) \, \mathrm{d}v \right].
\end{align}
Again, scaled versions of these residuals can be constructed as in \citet{BaddeleyEtal2005}. If $\hat{\rho}(C_1(u), \ldots, C_m(u))$ describes the intensity function of $X$ well, meaning e.g. that no relevant covariate was left out, $\tilde{s}(u)$ is expected to fluctuate around 0. 
Figure~\ref{fig:nonpar_residuals} illustrates that $\hat{\rho}$ is capable of capturing the correct form of dependence even without specifying a parametric model.

\begin{figure}[t]
    \centering
    \includegraphics[width=\textwidth]{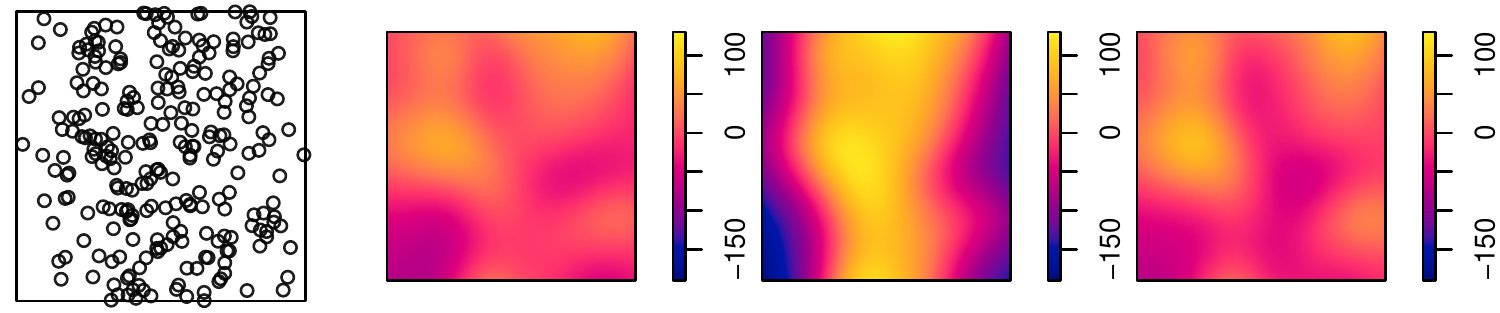}
    \caption{Left to right: realization of the Poisson process on $[0,1]^2$ with intensity function $\lambda(x,y) = 400 ( 1 - 4 (x-1/2)^2)$, the nonparametric smoothed residual field $\tilde{s}$ from \eqref{eq:nonparSRF} depending on the covariate $x$, the parametric smoothed residual field $s$ from \eqref{eq:paramSRF} with log-linear model depending on $x$, the parametric smoothed residual field $s$ from \eqref{eq:paramSRF} with log-linear model depending on $x$ and $x^2$.}
    \label{fig:nonpar_residuals}
\end{figure}

\subsection{Correlation coefficient between a point process and a covariate}\label{subsec:corr_coef}

Assume now that no nuisance covariates are given ($m=0$) and we want to investigate the strength of dependence between the intensity function of $X$ and a given covariate $C_1$. Without incorporating a possible effect of $C_1$, the natural estimate of the intensity function is constant, $\hat \lambda = X(W) / |W|$, and the smoothed residual field becomes
\begin{align}\label{eq:nonparSRFconst}
    \tilde{s}(u) = \frac{1}{e(u)} \sum_{x_i \in X \cap W} k(u-x_i) - \hat\lambda.
\end{align}
If the covariate $C_1$ does not influence $X$, we expect $C_1$ and $\tilde{s}$ to be independent. On the other hand, if $C_1$ influences the intensity function of $X$, $\tilde{s}$ should capture the dependence structure and exhibit correlations with $C_1$. This motivates us to quantify the strength of dependence between $X$ and $C_1$ by some measure of dependence between the two random fields $C_1$ and $\tilde{s}$.

To this end, we consider Kendall's correlation coefficient \citep[p.158]{Nelsen2006} and let $U_1, U_2$ be independent random vectors with uniform distribution in $W$. Denoting $Y=C_1(U_1) - C_1(U_2)$ and $Z=\tilde{s}(U_1)-\tilde{s}(U_2)$, we define
\begin{align}\label{eq:tau_population}
    \tau = & \mathbb{P}(Y \cdot Z > 0) - \mathbb{P}(Y \cdot Z < 0).
\end{align}
The empirical estimate of $\tau$ can be easily obtained if we consider a set of sampling points $\{ y_1, \ldots, y_n \}$, independently and uniformly distributed in $W$, independent of $X$ and $C_1$:
\begin{align}\label{eq:tauhat}
    \hat\tau = \frac{1}{n(n-1)} \sum_{i \neq j}\mathrm{sgn}(C_1(y_i) - C_1(y_j)) \; \mathrm{sgn}(\tilde{s}(y_i) - \tilde{s}(y_j)),
\end{align}
where $\mathrm{sgn}$ is the sign function. Naturally, the values of the correlation coefficient are restricted to the interval $[-1,1]$ and allow direct comparison of the strength of dependence between different data sets.

To illustrate the use of this correlation coefficient in quantifying the strength of dependence between a point process and a covariate, we perform the following experiment. We consider the Poisson process with the intensity function proportional to $\exp\{a x\}$ in the observation window $W=[0,1]^2$, for a given value of $a \in \mathbb{R}$, and with the expected number of points in $W$ fixed at 200. The covariate of interest is $C_1((x,y))=x$. The smoothed residual field $\tilde{s}$ from \eqref{eq:nonparSRFconst} is obtained with a large bandwidth $bw=0.5$ which reflects the fact that the true intensity function of the point process is very smooth. The value of $\hat\tau$ is then computed according to \eqref{eq:tauhat}. This is repeated for 500 independent realizations of the point process for each value of $a$ from a fine grid, and the means of $\hat\tau$ are plotted as a function of $a$ in the left panel of Figure~\ref{fig:kor_koef}. The plot shows that $\hat\tau$ increases in the absolute value with increasing strength of dependence, from 0 in case of independence ($a=0$) to almost 1 or -1 in case of very strong dependence. It also correctly captures the form of dependence (positive or negative association).

\begin{figure}[t]
    \centering
    \includegraphics[width=\textwidth]{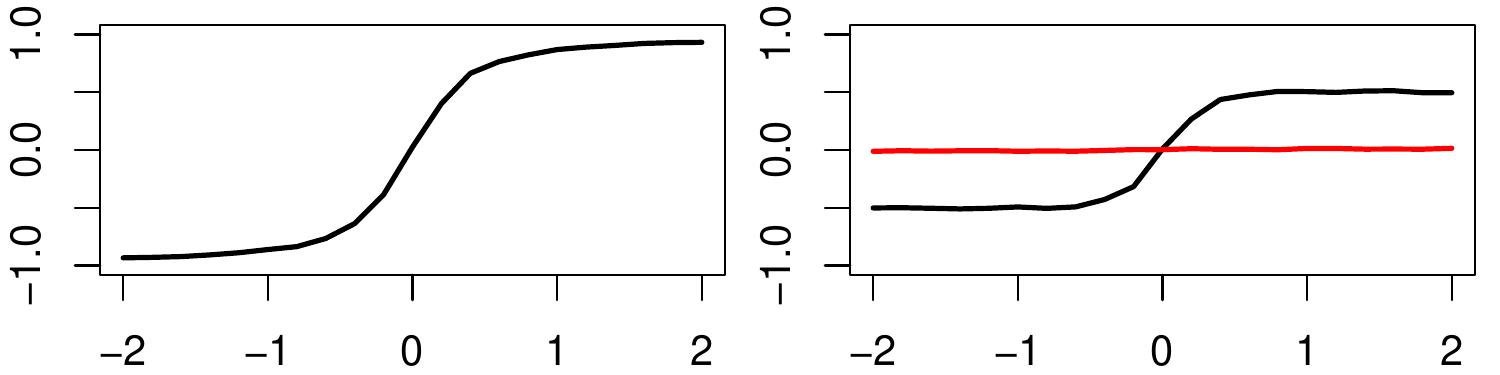}
    \caption{Left: plot of the correlation coefficient $\hat \tau$ as a function of the parameter $a$ for the example in Section~\ref{subsec:corr_coef}. Right: plot of the correlation coefficient $\hat \tau$ (black curve) and the partial correlation coefficient $\hat \tau_p$ (red curve) as functions of the parameter $a$ for the example in Section~\ref{subsec:partial_corr_coef}.}
    \label{fig:kor_koef}
\end{figure}

\subsubsection{Choice of sampling points}

We stress that independent sampling points need to be used in this case instead of simply using the observed points of $X \cap W$. In the latter case, the preferential sampling issues could arise, resulting in biased estimates of the properties of the two random fields \citep{Diggle2010a,Dvorak2022}. Loosely speaking, if, for example, the sampling points $\{ y_1, \ldots, y_n \}$ are more likely to be chosen in locations with high values of $C_1$, the sample mean and sample variance of $C_1(y_1), \ldots, C_1(y_n)$ do not reflect well the true properties of $C_1$. This negatively affects all subsequent steps of the analysis.

\subsubsection{Choice of measure of dependence}

Although different measures of dependence such as Pearson's or Spearman's correlation coefficients can be used, we suggest Kendall's correlation coefficient. It aligns well with the nonparametric spirit of this paper and has shown better performance in pre\-li\-mi\-na\-ry experiments not reported here and in previous studies on related topics \citep{Dvorak2022}.

\subsubsection{Choice of bandwidth}
For the construction of the smoothed residual field $\tilde{s}(u)$ in \eqref{eq:nonparSRFconst} one has to select a specific kernel function $k$ (a probability density function). The type of the kernel does not play an important role, and we use the Gaussian kernel. 
On the other hand, the choice of bandwidth (standard deviation of the kernel function) affects the properties of the estimates to a great extent. Traditional rules of thumb or more involved methods 
may be used for bandwidth selection in this case, see \citet[Section 6.5.1.2]{baddeley2015spatialR} or \citet{Cronie2018}. However, whenever available, expert knowledge about the specific problem at hand should guide the choice of bandwidth.

\subsection{Partial correlation coefficient between a point process and a covariate}\label{subsec:partial_corr_coef}

When several possibly correlated covariates are available, one might be interested in assessing the strength of dependence between the point process $X$ and the covariate of interest $C_{m+1}$ after removing the possible influence of the remaining (nuisance) covariates $C_1, \ldots, C_m$, in the spirit of the partial correlation coefficient.

The strength of dependence can be quantified by some measure of dependence between the covariate of interest $C_{m+1}$ and the smoothed residual field $\tilde{s}$ from \eqref{eq:nonparSRF} where the possible influence of the nuisance covariates $C_1, \ldots, C_m$ on $X$ has been removed. When a parametric model for the intensity function of $X$ is available, parametric residuals \eqref{eq:paramSRF} may be used instead.

We suggest using Kendall's correlation coefficient to quantify the dependence. Again, we consider a set of sampling points $\{ y_1, \ldots, y_n \}$, independently and uniformly distributed in $W$, independent of $X$ and $C_1, \ldots, C_{m+1}$, and define the sample version of the partial correlation coefficient as
\begin{align}\label{eq:tauhatpartial}
    \hat\tau_p = \frac{1}{n(n-1)} \sum_{i \neq j}\mathrm{sgn}(C_{m+1}(y_i) - C_{m+1}(y_j)) \; \mathrm{sgn}(\tilde{s}(y_i) - \tilde{s}(y_j)).
\end{align}
The population version can be defined in a similar way as in \eqref{eq:tau_population}. Concerning the choice of the sampling points and the choice of the measure of dependence, comments from the previous section apply here, too.

To illustrate the use of the partial correlation coefficient in quantifying the strength of dependence between a point process and a covariate of interest, after removing the influence of nuisance covariates, we performed the following experiment. The point process model is the Poisson process from Section~\ref{subsec:corr_coef}. Its intensity function depends in a log-linear way on the covariate $C_1((x,y))=x$, now treated as a nuisance covariate. Specifically, the intensity function is proportional to $\exp\{a x\}$. The covariate of interest is $C_2((x,y))=x+y$. We consider 500 independent realizations of the point process for each value of $a$ and compute the means of $\hat \tau$ and $\hat \tau_p$. The correlation coefficient $\hat\tau$, again computed with $bw=0.5$, correctly indicates that the point process depends on the covariate $C_2$ through the $x-$coordinate (black curve in the right panel of Figure~\ref{fig:kor_koef}). On the other hand, the partial correlation coefficient $\hat\tau_p$, computed with the adaptive choice of bandwidth described below, attains approx. zero values in this case (red curve in the right panel of Figure~\ref{fig:kor_koef}), implying that the influence of the nuisance covariate $C_1$ was successfully removed.

\subsubsection{Choice of bandwidth}\label{subsubsec:adaptive_bw}
Construction of the smoothed residual field requires choosing a bandwidth for the smoothing kernel. Again, standard recommendations may be employed, or the available expert knowledge may be utilized. However, in our pilot experiments with a single nuisance covariate $C_1$, we observed that the influence of $C_1$ was usually not completely removed from $X$ during the construction of the smoothed residual field $\tilde s(u)$, in the sense that the empirical Kendall's correlation coefficient of $\{ (\tilde{s}(y_j),C_1(y_j)), j = 1, \ldots, n \}$ was nonzero. Its value was strongly influenced by the value of bandwidth.

To remove this effect, we suggest selecting the bandwidth value (from a given finite set of candidate values) that minimizes the absolute value of the empirical Kendall's coefficient of $\{ (\tilde{s}(y_j),C_1(y_j)), j = 1, \ldots, n \}$, denoted $\hat \tau (\tilde{s},C_1)$ in the following. In this way, we select the bandwidth value that removes the influence of $C_1$ on $X$ the most successfully and it can be seen as a conservative version of the correlation coefficient. This is important mostly in cases where the nuisance covariate is correlated with the covariate of interest. For independent covariates, this procedure has very little effect on the performance of the random shift tests. We apply this approach to bandwidth selection in our simulation experiments below. When more than one nuisance covariate is available, this adaptive bandwidth procedure can be generalized by minimizing
$\sum_{i=1}^m \hat{\tau}(\tilde{s},C_i)^2$.


\subsection{Covariate-weighted residual measure}
While $\hat\tau_p$ is useful for quantifying the strength of dependence between $X$ and the covariate of interest $C_{m+1}$ after removing the influence of nuisance covariates $C_1, \ldots, C_m$, the random shift test using $\hat\tau_p$ as the test statistic turned out to have a rather low power in our simulation studies. The reason lies in the applied smoothing and the deliberate removal of the preferential sampling effects -- the association between the points of $X$ and the covariate $C_{m+1}$ brings important information.

To overcome these issues, we define the following characteristic that we call the \emph{covariate-weighted residual measure of $W$}:
\begin{align}\label{eq:CWR}
    CWR = \int_W C_{m+1}(u) \tilde{\mathcal{R}}(\mathrm{d}u) = & \sum_{x \in X \cap W} C_{m+1}(x) - \int_W C_{m+1}(u) \hat{\rho}(C_1(u), \ldots, C_m(u)) \, \mathrm{d}u.
\end{align}
This can be viewed as a generalization of the test statistic $T$ from Section~\ref{subsec:PCtest} which also includes the sum of covariate values, but does not take into account possible nuisance covariates. By sampling the values of $C_{m+1}$ at the points of $X$ we take advantage of any possible preferential sampling effects, and no smoothing is performed when computing the value of $CWR$, hence we avoid the problem of bandwidth selection. The expectation of $CWR$ is close to 0 if the covariates $C_1, \ldots, C_m$ capture all variation in $\lambda(u)$, i.e. if $\hat \rho$ is close to $\lambda$, and will differ from 0 otherwise. This enables testing the significance of $C_{m+1}$ after removing the influence of $C_1, \ldots, C_m$.


\subsection{Testing the covariate significance under the presence of nuisance covariates}
Now we focus on the null hypothesis that $X$ and $C_{m+1}$ are independent, conditionally on $C_1, \ldots, C_m$. We employ the random shift test described in Section~\ref{sec:RStests}, either with torus or variance correction. The test statistic can be $\hat\tau_p$ in which case the two spatial objects to be shifted against each other are the two random fields $\Phi = \tilde{s}$ and $\Psi = C_{m+1}$. Alternatively, one can use the covariate-weighted residual measure of $W$ as the test statistic. In this case $\Phi = \tilde{R}$ is a measure and $\Psi = C_{m+1}$ is a random field. If $v_i$ is a shift vector, the shift of the random field $\Psi$ should be interpreted in both cases as $(\Psi + v_i)(u) = \Psi(u - v_i)$. The choice of the correction factors for the variance correction is discussed in Appendix~\ref{appendix:correction_factors}, including Proposition~1 which studies the variance of $CWR$ for Poisson processes and an empirical study for log-Gaussian Cox processes. The assumption of stationarity of one of the spatial objects is discussed in detail in Section~\ref{sec:CD}.

\section{Simulation study}\label{sec:simulations}




To assess the performance of the proposed tests, we present below a set of simulation experiments, both under the null hypothesis and under various alternatives. The models range from clustering through complete spatial randomness to regularity, even combining clustering and inhibition on different scales. The null hypothesis states that $X$ and $C_{m+1}$ are independent, conditionally on the nuisance covariates. For simplicity, we focus on the situation with a single nuisance covariate. The proposed nonparametric tests are compared with the parametric methods available in standard software represented by the \texttt{spatstat} package.

\subsection{Simulation study design}\label{subsec:simulation_design}

The following notation and choices are used in all simulation experiments. $Z_1, Z_2, \ldots$ are independent, identically distributed Gaussian random fields, centered, unit variance, with exponential covariance function with scale 0.1. The observation window is $W=[0,1]^2$. The expected number of points in $W$ is equal to $\exp\{ 5 \} \doteq 148.4$ for Poisson and clustered models and approximately equal to $\exp\{ 5 \}$ for models exhibiting regularity.

For each model, we simulate $5\,000$ independent realizations, and for each realization, we perform a set of tests on the 5\% nominal significance level. In the tables of results we report the fractions of rejections for the individual tests, rounded to three decimals. To assess the liberality or conservativeness of the tests, one can compare the reported rejection rates (in experiments performed under the null hypothesis) with the interval based on the 2.5~\% and 97.5~\% quantiles of the binomial distribution with parameters $n=5\,000$ and $p=0.05$, that is, with the interval $[0.0440, 0.0562]$.

We investigate the performance of the random shift tests with either $\hat\tau_p$ or $CWR$ as the test statistic, with either parametric or nonparametric version of residuals (denoted by the symbol ``p'' or ``n'' in the tables of results) and with either torus or variance correction (denoted by ``tor'' or ``var'' in the tables of results). The values of $\hat\tau_p$ are computed with the bandwidth selected by the adaptive procedure from Section~\ref{subsubsec:adaptive_bw} and with 100 sampling points chosen uniformly and independently in $W$. The random shift tests are compared with the parametric tests provided by the functions \texttt{ppm} (for Poisson, Strauss, and hardcore Strauss processes) and \texttt{kppm} (for log-Gaussian Cox processes, denoted LGCP in the following) from the \texttt{spatstat} package, see Section~\ref{subsec:parametric_methods}.

To mimic the practical issues with model specification, we consider these parametric tests both with the correct interaction model and with an incorrect interaction model of a similar type. Specifically, in addition to fitting the correct model to the LGCPs we also fit a Matérn cluster process; for the Strauss and hardcore Strauss processes we fit the models with the interaction distance fixed to either the correct or incorrect value, specified in the tables of results in the column ``Variant''. We also fit an inhomogeneous Poisson process to all data sets to investigate the effect of ignoring the interaction structure. On the other hand, we do not try fitting clustered models to clearly regular data sets and vice versa.

All the parametric tests assume the log-linear model for the intensity function \eqref{eq:lambda_loglinear}, even though for some point process models we consider below this does not hold. This also illustrates possible issues with model misspecification in practice.

\subsection{Significance level under independent covariates}\label{subsec:sim_size_independent}


In the following, we let the nuisance covariate influencing the intensity function of the point process be $C_1(u) = Z_1(u)$ and the covariate of interest be $C_2(u) = Z_3(u)$, which means that the covariate of interest $C_2$ is independent of the nuisance covariate $C_1$ and the point process $X$. For the construction of the LGCP models we also use the random field $Z_2$, which is responsible for interactions in the point process rather than variation in its intensity function. For the Poisson and LGCP models, the covariate $C_1$ influences the intensity function directly. For the Strauss and hardcore Strauss models, it directly influences the trend function $\beta(u)$. This influence is transformed to the intensity function in a nontrivial way. We consider the following models:
\begin{itemize}
    \item[($P_1$)] Poisson process with intensity function $\lambda(u) = \exp \{ 4.5 + Z_1(u) \}$ 
    \item[($P_2$)] Poisson process with intensity function $\lambda(u) = \exp \{ 5 \} \cdot Z_1(u)^2$ 
    \item[($L_1$)] LGCP with driving intensity function $\Lambda(u) = \exp \{ 4.0 + Z_1(u) + Z_2(u) \}$ 
    \item[($L_2$)] LGCP with driving intensity function $\Lambda(u) = \exp \{ 4.5 + Z_2(u) \} \cdot Z_1(u)^2$
    \item[($S_1$)] Strauss process with interaction parameter $\gamma = 0.5$, interaction range $R = 0.05$ and trend $\beta(u) = 220 \cdot \exp\{Z_1(u)\}$
    \item[($S_2$)] Strauss process with $\gamma = 0.5, R = 0.05$ and $\beta(u) = 350 \cdot Z_1(u)^2$
    \item[($H_1$)] Strauss process with hardcore distance $hc=0.01$, interaction parameter $\gamma = 4$, interaction distance $R=0.02$ and trend $\beta(u) = 180 \cdot \tilde{Z}(u)$, where $\tilde{Z}(u) = c \cdot \exp\{ Z_1(u)/5 \}$ and $c$ is chosen for each realization so that the maximum of the given realization of $\tilde{Z}(u)$ over $W$ is 1.
    \item[($H_2$)] Strauss process with hardcore distance $hc=0.01$, $\gamma = 4$, $R=0.02$ and $\beta(u) = 120 \cdot \tilde{Z}(u)$, where $\tilde{Z}(u) = c \cdot \max(1-Z_1(u)^2/5,0)$, again scaled for each realization to attain the maximum value of 1 over $W$.
\end{itemize}
Note that in the shorthand notation for the models the letter represents the type of interaction in the point process while the subscript specifies whether the covariate $C_1$ influences the intensity function in a log-linear way (denoted by 1) or in a quadratic way (denoted by 2).

Since the covariate of interest $C_2$ is independent of $X$, the tests should reject in 5~\% of cases. Table~\ref{tab:simulace_size_independent} shows the fractions of rejection. We make the following observations:
\begin{itemize}
    \item The nonparametric tests match the nominal significance level correctly for all models, the tests based on $CWR$ match it slightly more precisely than those based on $\hat\tau_p$. Both the torus correction and the variance correction perform well, with only a slight tendency toward liberality observed for the torus correction and the tests based on $\hat\tau_p$.
    \item Parametric tests assuming correct interaction structure and correct model for the intensity function (denoted by 1) match the nominal significance level correctly for the Poisson process (P) while being highly liberal for the LGCP (L) and hardcore Strauss process (H). They are slightly conservative for the Strauss process (S). 
    \item Parametric tests assuming correct interaction structure and incorrect model for the intensity function (denoted by 2) may exhibit very strong liberality (P, H) or conservativeness (L).
    \item Parametric tests assuming incorrect interaction structure may exhibit very strong liberality (e.g. assuming Poisson interactions for L or H models) or conservativeness (e.g. assuming attractive interactions for P models).
\end{itemize}
These observations illustrate that parametric tests are prone to perform poorly under model misspecification either in terms of the interaction structure or the intensity function. However, even when both of these model components are specified correctly, there is a risk of strong liberality of the parametric tests with the sample sizes considered here. From this point of view, the nonparametric tests are preferable as they match the nominal significance level correctly for all models in this study.

\begin{table}
\caption{Size of the tests, independent covariates -- fractions of rejection. For the H models the asterisk signifies that the correct hardcore distance was assumed in the given parametric test, whereas for the P and S models no hardcore distance is assumed.
\label{tab:simulace_size_independent}}
\renewcommand{\arraystretch}{1.2}
\begin{tabular}{|l|l||c|c||c|c||c|c||c|c|}
 \hline
 Test   & Variant   & $P_1$    & $P_2$    & $L_1$    & $L_2$    & $S_1$ & $S_2$ & $H_1$ & $H_2$  \\ \hline\hline
 $\hat\tau_p$ & p, tor    & 0.053 & 0.058 & 0.059 & 0.065 & 0.056 & 0.062 & 0.063 & 0.061 \\ \hline
 $\hat\tau_p$ & p, var    & 0.040 & 0.043 & 0.031 & 0.047 & 0.043 & 0.049 & 0.047 & 0.045  \\ \hline
 $\hat\tau_p$ & n, tor    & 0.049 & 0.057 & 0.059 & 0.063 & 0.058 & 0.060 & 0.061 & 0.062  \\ \hline
 $\hat\tau_p$ & n, var    & 0.043 & 0.046 & 0.035 & 0.047 & 0.046 & 0.047 & 0.047 & 0.050  \\ \hline\hline
 CWR          & p, tor    & 0.049 & 0.053 & 0.055 & 0.051 & 0.048 & 0.051 & 0.057 & 0.048  \\ \hline
 CWR          & p, var    & 0.050 & 0.049 & 0.050 & 0.054 & 0.047 & 0.050 & 0.056 & 0.048  \\ \hline
 CWR          & n, tor    & 0.049 & 0.050 & 0.054 & 0.050 & 0.045 & 0.045 & 0.056 & 0.051  \\ \hline
 CWR          & n, var    & 0.048 & 0.051 & 0.049 & 0.048 & 0.048 & 0.045 & 0.055 & 0.049  \\ \hline\hline
 ppm          & Poisson   & 0.048 & 0.166 & 0.268 & 0.331 & 0.021 & 0.052 & 0.198 & 0.161  \\ \hline\hline
 kppm         & LGCP      & 0.020 & 0.023 & 0.080 & 0.027 & --    & --    & --    & --  \\ \hline
 kppm         & MC        & 0.041 & 0.025 & 0.086 & 0.030 & --    & --    & --    & --  \\ \hline\hline
 ppm          & Str(0.02) & 0.039 & 0.111 & --    & --    & 0.019 & 0.040 & 0.134* & 0.137*  \\ \hline
 ppm          & Str(0.05) & 0.044 & 0.070 & --    & --    & 0.038 & 0.060 & 0.095* & 0.079*  \\ \hline
 ppm          & Str(0.10) & 0.041 & 0.063 & --    & --    & 0.028 & 0.034 & 0.093* & 0.092*  \\ \hline
\end{tabular}
\end{table}

\subsection{Significance level under dependent covariates}\label{subsec:sim_size_correlated}

In this section, we consider the case of the covariate of interest $C_2$ being correlated with the nuisance covariate $C_1$. We denote $C_1(u) = Z_1(u)$ and $C_2(u) = Z_1(u) + b Z_3(u)$ for different values of $b > 0$. A smaller value of $b$ implies a stronger correlation between the covariates. We consider the same point process models as in Section~\ref{subsec:sim_size_independent}. In fact, we use the same realizations and simply construct the covariate $C_1$ in a different way.

We have performed the simulation experiments for the $L_1$, $L_2$, $S_1$ and $S_2$ models. However, we report the results only for the $L_1$ model (denoted $L_1^*$ below to indicate that the covariates are correlated) since the observations made in all these cases are the same.

Table~\ref{tab:simulace_size_dependent} shows the fractions of rejection for the $L_1^*$ model with different choices of $b$ together with the results for the original $L_1$ model from Section~\ref{subsec:sim_size_independent} which can be considered as the limiting case for $b \rightarrow \infty$. The parametric tests exhibit the same rejection rates no matter the value of $b$ due to the specific choice of the (linear) form of the covariates and the log-linear form of the intensity function -- for all values of $b$ the parametric tests in fact fit the same model by putting different weights on $C_1$ and $C_2$.

The nonparametric tests show an increasing level of conservativeness with the increasing correlation between the covariates (with $b$ going to 0). For $CWR$ this is caused by the nature of the random shift test and the preferential sampling effects which reduce the variance of the test statistic computed from the observed data (with no shift) compared to the test statistic values computed from shifted data (where the preferential sampling effects are reduced or removed completely), as confirmed by simulation experiments not reported here. Similar conservativeness also appears for $\hat\tau_p$. 
We consider this conservativeness to be a smaller issue than liberality and conclude that this observation does not provide arguments against the use of the nonparametric tests.

\begin{table}
\caption{Size of the tests, correlated covariates -- fractions of rejection for the $L_1^*$ model with different values of $b$.\label{tab:simulace_size_dependent}}
\renewcommand{\arraystretch}{1.2}
\begin{tabular}{|l|l||c|c|c|c|c||c|}
 \hline
 Test   & Variant  & $L_1^*$   & $L_1^*$   & $L_1^*$   & $L_1$ \\ \hline\hline
 $b$    &          & 1     & 2     & 4     & --  \\ \hline\hline
 $\hat\tau_p$ & p, tor   & 0.030 & 0.044 & 0.053 & 0.059   \\ \hline
 $\hat\tau_p$ & p, var   & 0.021 & 0.029 & 0.032 & 0.031   \\ \hline
 $\hat\tau_p$ & n, tor   & 0.028 & 0.042 & 0.051 & 0.059   \\ \hline
 $\hat\tau_p$ & n, var   & 0.018 & 0.026 & 0.031 & 0.035   \\ \hline\hline
 CWR          & p, tor   & 0.013 & 0.035 & 0.045 & 0.055  \\ \hline
 CWR          & p, var   & 0.014 & 0.033 & 0.042 & 0.050  \\ \hline
 CWR          & n, tor   & 0.017 & 0.039 & 0.047 & 0.054  \\ \hline
 CWR          & n, var   & 0.014 & 0.034 & 0.044 & 0.049  \\ \hline\hline
 ppm          & Poisson  & 0.268 & 0.268 & 0.268 & 0.268   \\ \hline\hline
 kppm         & LGCP     & 0.080 & 0.080 & 0.080 & 0.080   \\ \hline
 kppm         & MC       & 0.086 & 0.086 & 0.086 & 0.086   \\ \hline
\end{tabular}
\end{table}

\subsection{Power under dependent covariates}\label{subsec:sim_power_correlated}

In this section, we study the power of the tests in the situations where the covariate of interest $C_2$ influences the intensity function of $X$ even after removing the effect of the nuisance covariate $C_1$. We consider models similar to those in the previous sections and let $C_1(u) = Z_1(u)$ and $C_2(u) = Z_1(u) + 2 Z_3(u)$. We focus on the case with dependent covariates, which is more challenging for our proposed tests as they showed conservativeness in Section~\ref{subsec:sim_size_correlated}. The models depend on a parameter $a > 0$ that controls the strength of dependence between $X$ and the covariate of interest $C_2$. The value of $a$ is chosen so that all the tests exhibit nontrivial powers, i.e. not close to 0.05 and not close to 1.00. The models are given as follows:
\begin{itemize}
    \item[($P_1^p$)] Poisson process with intensity function $\lambda(u) = \exp \{ 4.5 + Z_1(u) + aZ_3(u) - a^2/2 \}$ with $a=1/4$.
    \item[($P_2^p$)] Poisson process with intensity function $\lambda(u) = \exp \{ 5.0 + aZ_3(u) - a^2/2 \} \cdot Z_1(u)^2$ with $a=1/4$.
    \item[($L_1^p$)] LGCP with driving intensity function $\Lambda(u) = \exp \{ 4.0 + Z_1(u) + Z_2(u) + aZ_3(u) - a^2/2 \}$ with $a=1/2$.
    \item[($L_2^p$)] LGCP with driving intensity function $\Lambda(u) = \exp \{ 4.5 + Z_2(u) + aZ_3(u) - a^2/2 \} \cdot Z_1(u)^2$ with $a=1/2$.
    \item[($S_1^p$)] Strauss process with interaction parameter $\gamma = 0.5$, interaction range $R = 0.05$ and trend $\beta(u) = 210 \cdot \exp\{Z_1(u) + a Z_3(u)\}$ with $a=1/4$.
    \item[($S_2^p$)] Strauss process with $\gamma = 0.5, R = 0.05$ and $\beta(u) = 350 \cdot \exp\{a Z_3(u)\} \cdot Z_1(u)^2$ with $a=1/4$.
    \item[($H_1^p$)] Strauss process with hardcore distance $hc=0.01$, interaction parameter $\gamma = 4$, interaction distance $R=0.02$ and trend $\beta(u) = 190 \cdot \tilde{Z}(u)$, where $\tilde{Z}(u) = c \cdot \exp\{ Z_1(u)/5 + aZ_3(u) - a^2/2\}$ and $c$ is chosen for each realization so that the maximum of the given realization of $\tilde{Z}(u)$ over $W$ is 1.
    \item[($H_2^p$)] Strauss process with hardcore distance $hc=0.01$, $\gamma = 4$, $R=0.02$ and $\beta(u) = 170 \cdot \tilde{Z}(u)$, where $\tilde{Z}(u) = c \cdot \exp\{ aZ_3(u) - a^2/2\} \cdot \max(1-Z_1(u)^2/5,0)$, again scaled for each realization to attain the maximum value of 1 over $W$.
\end{itemize}

Table~\ref{tab:simulace_power} shows the fractions of rejections for the individual tests for the eight models specified above. We make the following observations:
\begin{itemize}
    \item For both $\hat\tau_p$ and $CWR$, the versions of the test based on nonparametric residuals exhibit higher power than those based on parametric residuals.
    \item The tests based on $\hat\tau_p$ have very low power due to the smoothing and removal of the preferential sampling effects.
    \item The tests based on $CWR$ exhibit very high power comparable to the parametric tests with correct interaction model and correct model for the intensity function (for P, L, and H models) or even higher power (S).
    \item When the parametric tests are used with the correct interaction model and incorrect model for the intensity function, the nonparametric tests based on $CWR$ have much higher power (L, S), slightly higher power (H), or direct comparison is not possible due to severe liberality of the parametric test (P).
    \item The torus correction and the variance correction perform nearly equivalently for tests based on $CWR$, while for tests based on $\hat\tau_p$ the torus correction shows slightly higher power, which can be explained by the small liberality of these tests observed in Section~\ref{subsec:sim_size_independent}.
\end{itemize}
These observations indicate that the random shift tests based on $CWR$ with nonparametric residuals and either torus or variance correction can be preferred in practice to parametric tests since the possible issues with model misspecification are avoided without compromising the power of the test.

\begin{table}
\caption{Power of the tests, correlated covariates -- fractions of rejection. For the H models the asterisk signifies that the correct hardcore distance was assumed in the given parametric test, whereas for the P and S models no hardcore distance is assumed.\label{tab:simulace_power}}
\renewcommand{\arraystretch}{1.2}
\begin{tabular}{|l|l||c|c||c|c||c|c||c|c|}
 \hline
 Test   & Variant   & $P_1^p$ & $P_2^p$ & $L_1^p$ & $L_2^p$ & $S_1^p$ & $S_2^p$ & $H_1^p$   & $H_2^p$ \\ \hline\hline
 $\hat\tau_p$ & p, tor    & 0.198 & 0.126 & 0.153 & 0.155 & 0.127 & 0.114 & 0.101 & 0.128 \\ \hline
 $\hat\tau_p$ & p, var    & 0.160 & 0.105 & 0.113 & 0.121 & 0.100 & 0.096 & 0.076 & 0.100 \\ \hline
 $\hat\tau_p$ & n, tor    & 0.215 & 0.214 & 0.164 & 0.199 & 0.137 & 0.152 & 0.104 & 0.134 \\ \hline
 $\hat\tau_p$ & n, var    & 0.177 & 0.178 & 0.125 & 0.164 & 0.111 & 0.122 & 0.080 & 0.109 \\ \hline\hline
 CWR          & p, tor    & 0.758 & 0.444 & 0.820 & 0.702 & 0.608 & 0.396 & 0.310 & 0.405 \\ \hline
 CWR          & p, var    & 0.753 & 0.441 & 0.808 & 0.698 & 0.604 & 0.392 & 0.304 & 0.402 \\ \hline
 CWR          & n, tor    & 0.793 & 0.677 & 0.846 & 0.791 & 0.690 & 0.514 & 0.331 & 0.418 \\ \hline
 CWR          & n, var    & 0.789 & 0.675 & 0.835 & 0.783 & 0.688 & 0.511 & 0.324 & 0.416 \\ \hline\hline
 ppm          & Poisson   & 0.774 & 0.720 & 0.960 & 0.943 & 0.507 & 0.483 & 0.618 & 0.714 \\ \hline\hline
 kppm         & LGCP      & 0.644 & 0.254 & 0.841 & 0.477 & --    & --    & --     & -- \\ \hline
 kppm         & MC        & 0.745 & 0.264 & 0.848 & 0.499 & --    & --    & --     & -- \\ \hline\hline
 ppm          & Str(0.02) & 0.682 & 0.579 & --    & --    & 0.432 & 0.381 & 0.346* & 0.375* \\ \hline
 ppm          & Str(0.05) & 0.614 & 0.456 & --    & --    & 0.467 & 0.399 & 0.302* & 0.340* \\ \hline
 ppm          & Str(0.10) & 0.515 & 0.387 & --    & --    & 0.344 & 0.263 & 0.307* & 0.355* \\ \hline
\end{tabular}
\end{table}

\subsection{Results of further simulation experiments}

In the following, we comment on some observations made from further simulation experiments not reported here. First, the scaled versions of the residuals discussed in \citet{BaddeleyEtal2005} can be used instead of the raw residuals \eqref{eq:paramR} and \eqref{eq:nonparR}. We have investigated the performance of the nonparametric tests based on the inverse and Pearson residuals and compared it to the performance of the tests based on the raw residuals. In terms of rejection rates under the null hypothesis, we have found no significant differences between the three types of residuals. Concerning the power of the tests, the raw and Pearson residuals performed equally well, while the inverse residuals exhibited somewhat smaller power.

Second, the nonparametric estimation of the intensity function depending on a covariate can be performed by the \texttt{spatstat} function \texttt{rhohat} using three types of estimators: ``ratio'', ``reweight'' and  ``transform'' \citep{Baddeley2012}. In Sections \ref{subsec:sim_size_independent} to \ref{subsec:sim_power_correlated} we reported the results for the default ratio estimator. The rejection rates for the ratio and the reweight estimators were comparable, while being somewhat higher for the transform estimator, showing slight liberality under the null hypothesis accompanied by slightly higher power under the alternatives.

Finally, the random shift tests were performed in the previous sections with shift vectors generated from the uniform distribution on a disc both for the torus correction and the variance correction to enable direct comparison. When the observation window $W$ is rectangular and the torus correction is used, it might be more natural to consider the shift vectors generated from the uniform distribution on the whole $W$. In a smaller simulation experiment, the two versions of the random shift test with torus correction performed similarly, with a small tendency towards liberality for the version with shift vectors generated uniformly on $W$.

\section{Nonparametric covariate selection for the BCI data set}\label{sec:application_BCI}

To illustrate the possibility to use the proposed random shift tests for covariate selection, we consider now the BCI data set described in Section~\ref{subsec:motivation}. Five covariates are available that possibly influence the intensity function of the point process. A possible way to select the set of covariates that have a significant effect on the intensity function is the backward selection procedure described in the following. The numerical results are given in Table~\ref{tab:results_BCI}.

We start in stage 1 with all five covariates, and for each of those we perform the random shift test where the given covariate is the covariate of interest and the remaining four covariates are considered to be the nuisance covariates. We use the test based on $CWR$ with nonparametric residuals and torus correction, with 999 random shifts where the shift vectors have uniform distribution on a disc with radius 250 metres. The covariate with the highest p-value (potassium in this case, printed in italics in Table~\ref{tab:results_BCI}) is removed, and the procedure is repeated in stage 2 with the four covariates. In this stage, the nitrogen covariate is removed, then the gradient covariate, and finally in stage 4 where only two covariates are considered (elevation, phosphorus), both covariates are found significant on the 5\% significance level, see Table~\ref{tab:results_BCI}. We conclude that these two covariates significantly affect the intensity function of the point process and should be included in the further steps of the inference. Other covariates can be disregarded without losing important information.

For comparison, we have also fitted the log-linear model \eqref{eq:lambda_loglinear} with the five covariates considered here, using the \texttt{kppm} function from the \texttt{spatstat} package as described in Section~\ref{subsec:parametric_methods}. We assume the Thomas type of interactions as suggested in \citet[Sec.~12.4.4]{baddeley2015spatialR}. With this approach, three covariates are found significant on the 5\% significance level: elevation, gradient and phosphorus, with p-values 0.019, 0.044 and $10^{-4}$, respectively. Two of these covariates were also found significant by the nonparametric procedure described above (elevation and phosphorus, see Table~\ref{tab:results_BCI}). On the other hand, the gradient covariate was found borderline significant by the parametric approach and not significant by the nonparametric procedure.


\begin{table}
\caption{Backward selection of covariates for the BCI data set. Individual cells show the p-values of the random shift tests. The row indicates the covariate of interest, while all other covariates considered in the given stage (column) are considered to be the nuisance covariates.\label{tab:results_BCI}}
    \renewcommand{\arraystretch}{1.2}
	\centering
		\begin{tabular}{|l||c|c|c|c|}
		\hline
			           & stage 1 & stage 2 & stage 3 & stage 4  \\ \hline\hline
			elevation  & 0.168   & 0.186   & 0.018   & 0.016   \\ \hline
			gradient   & 0.228   & 0.240   & \emph{0.232}   & --      \\ \hline
			nitrogen   & 0.426   & 0.\emph{388}   & --      & --      \\ \hline
			phosphorus & 0.066   & 0.128   & 0.220   & 0.048   \\ \hline
			potassium  & \emph{0.598}   & --      & --      & --      \\ \hline
		\end{tabular}
\end{table}

\section{Nonparametric comparison of dependence strength for the CLM data set} \label{sec:application_CLM}

We now focus on the CLM data set described in Section~\ref{subsec:motivation}. To assess the strength of dependence of the forest fire locations on the two available covariates we estimate the correlation coefficient $\tau$ using \eqref{eq:tauhat}. The bandwidth is chosen as the default value from the \texttt{spatstat} function \texttt{density.ppp} which is 50 kilometers in this case. For the elevation covariate the estimated value is 0.035, while for the gradient covariate it is 0.103. The positive signs of the estimated values indicate that the intensity of point occurrence tends to be higher in locations with high covariate values. However, the influence of elevation seems to be negligible and the influence of gradient appears to be very small. When looking at the partial correlation coefficients $\hat \tau_p$ from \eqref{eq:tauhatpartial}, removing the influence of the other covariate, we obtain the value 0.031 for the elevation covariate and 0.103 for the gradient covariate, respectively.

We may also estimate the correlation coefficients for the BCI data set and compare the strength of dependence between the point process and the covariates (elevation, gradient) between the two data sets (BCI vs. CLM). For the BCI data set we choose the bandwidth of 62.5 kilometers in the same way as above. For the elevation covariate the estimated value of $\tau$ is -0.048, for the gradient covariate it is 0.249. When removing the influence of all the remaining covariates, including the soil mineral contents, the estimated value of $\tau_p$ for the elevation covariate is 0.083, for the gradient covariate it is 0.172. We conclude that the influence of the gradient on the point process, after removing the influence of the remaining available covariates, is nearly twice as strong in the BCI data set than in the CLM data set. The influence of elevation, as quantified by the correlation and partial correlation coefficients, is much weaker than the influence of gradient in both data sets.

We remark that in Section~\ref{sec:application_BCI}, the elevation covariate has been determined to have a stronger influence on the point process in the BCI dataset than the gradient covariate while the opposite observation has been made in this section. This can be attributed to the strong dependence between gradient and elevation and the conceptually different methods applied: the $CWR$ test statistic uses the covariate values directly while the correlation coefficients only use the signs of the differences of covariate values. Also, smoothing is required for computation of the correlation coefficients while it is avoided for $CWR$.

\section{Conclusions and discussion}\label{sec:CD}
The methods proposed in this paper allow quantification and testing of the significance of the correlations between a point process and a covariate of interest, possibly after removing the influence of nuisance covariates. We stress that the proposed methods can be used without specifying any model for the data. The simulation experiments reported in Section~\ref{sec:simulations} show that the random shift tests based on $\hat\tau_p$ or $CWR$ match the nominal significance level correctly even in situations where parametric tests based on asymptotic distributions (assuming the correct form of interactions in the point process and the correct form of the intensity function) exhibit different degrees of liberality or conservativeness. Under model misspecification, the parametric tests may suffer from even more severe problems. Concerning power, the nonparametric tests based on $CWR$ exhibit comparable or even higher power than parametric tests under the correct model, while showing higher power than parametric tests under incorrect models (where either the interaction or the intensity function is misspecified). This indicates the superiority of the $CWR$ tests over parametric tests in practical applications where the true model is not known. Hence, using the proposed nonparametric $CWR$ tests for covariate selection, e.g. as discussed in Section~\ref{sec:application_BCI}, provides more reliable results than the available parametric tests are able to provide, and the selected covariates can be used in the further steps of inference with greater confidence.

The only assumption of the proposed random shift tests is that at least one of the objects is stationary under the null hypothesis, so that its distribution is not affected by the shifts. Either the covariate of interest can be assumed to be stationary or the covariate-weighted residual measure or the smoothed residual field can be assumed to be close to stationarity if all the relevant covariates are used in construction of the residuals.

A natural question is whether the proposed methods are applicable also to categorical covariates. If one of the nuisance covariates is categorical, nonparametric estimation of the intensity function may be performed separately on the individual subregions of $W$ determined by the categorical covariate, allowing all the proposed methods to be used as described above. If the categorical covariate is the covariate of interest, computing $\hat\tau$ or $\hat\tau_p$ is not relevant due to the ties in the data. 
However, the observation window $W$ can be separated into subregions $W_1, \ldots, W_k$ determined by the values of the covariate of interest. The values $V_i = n(X \cap W_i) - \int_{W_i} \hat \rho(C_1(u), \ldots, C_m(u)) \, \mathrm{d} u, \; i = 1, \ldots, k,$ can be used to form a vector test statistic $(V_1, \ldots, V_k)$ and the random shift test can be performed e.g. by means of the global envelope test \citep{MyllymakiEtal2017}. This approach corresponds to the determination of differences between point process intensities in subregions $W_1, \ldots, W_k$.

\section*{Acknowledgements}
We are grateful to Jim Dalling, Robert John, Kyle Harms, Robert Stallard, and Joe Yavitt for collecting the BCI soil data (supported by STRI Soils Initiative and CTFS and the grants NSF DEB021104, 021115, 0212284, 0212818, OISE 0314581, with field assistance by Paolo Segre and Juan Di Trani).

\bibliography{Tomas_Bibfile}

\begin{thebibliography}{26}
\expandafter\ifx\csname natexlab\endcsname\relax\def\natexlab#1{#1}\fi
\expandafter\ifx\csname url\endcsname\relax
  \def\url#1{\texttt{#1}}\fi
\expandafter\ifx\csname urlprefix\endcsname\relax\def\urlprefix{URL: }\fi

\bibitem[{Baddeley et~al.(2012)Baddeley, Chang, Song and Turner}]{Baddeley2012}
Baddeley, A., Chang, Y.-M., Song, Y. and Turner, R. (2012) Nonparametric
  estimation of the dependence of a spatial point process on spatial
  covariates.
\newblock \textit{Stat. Interface}, \textbf{5}, 221--236.

\bibitem[{Baddeley et~al.(2015)Baddeley, Rubak and
  Turner}]{baddeley2015spatialR}
Baddeley, A., Rubak, E. and Turner, R. (2015) \textit{Spatial Point Patterns:
  Methodology and Applications with R}.
\newblock Chapman \& Hall Interdisciplinary Statistics Series. CRC Press, Boca
  Raton, Florida.

\bibitem[{Baddeley and Turner(2005)}]{Baddeley2005a}
Baddeley, A. and Turner, R. (2005) spatstat: An {R} package for analyzing
  spatial point patterns.
\newblock \textit{J. Stat. Softw.}, \textbf{12}, 1--42.

\bibitem[{Baddeley et~al.(2005)Baddeley, Turner, M{\o}ller and
  Hazelton}]{BaddeleyEtal2005}
Baddeley, A., Turner, R., M{\o}ller, J. and Hazelton, M. (2005) Residual
  analysis for spatial point processes (with discussion).
\newblock \textit{J. R. Stat. Soc. Ser. B Stat. Methodol.}, \textbf{67},
  617--666.

\bibitem[{Choiruddin et~al.(2021)Choiruddin, Coeurjolly and
  Waagepetersen}]{Choiruddin2021}
Choiruddin, A., Coeurjolly, J.-F. and Waagepetersen, R. (2021) Information
  criteria for inhomogeneous spatial point processes.
\newblock \textit{Aust. N. Z. J. Stat.}, \textbf{63}, 119--143.

\bibitem[{Coeurjolly and Lavancier(2013)}]{Coeurjolly2013}
Coeurjolly, J.-F. and Lavancier, F. (2013) Residuals and goodness-of-fit tests
  for stationary marked gibbs point processes.
\newblock \textit{J. R. Stat. Soc. Ser. B Stat. Methodol.}, \textbf{75},
  247--276.

\bibitem[{Coeurjolly and Rubak(2013)}]{CoeurjollyRubak2013}
Coeurjolly, J.-F. and Rubak, E. (2013) Fast covariance estimation for
  innovations computed from a spatial {G}ibbs point process.
\newblock \textit{Scand. J. Stat.}, \textbf{40}, 669--684.

\bibitem[{Condit(1998)}]{Condit1998}
Condit, R. (1998) \textit{Tropical Forest Census Plots}.
\newblock Springer-Verlag and R. G. Landes Company.

\bibitem[{Cronie and Van~Lieshout(2018)}]{Cronie2018}
Cronie, O. and Van~Lieshout, M. N.~M. (2018) {A non-model-based approach to
  bandwidth selection for kernel estimators of spatial intensity functions}.
\newblock \textit{Biometrika}, \textbf{105}, 455--462.

\bibitem[{Dale and Fortin(2002)}]{DaleFortin2002}
Dale, M. R.~T. and Fortin, M.-J. (2002) Spatial autocorrelation and statistical
  tests in ecology.
\newblock \textit{Ecoscience}, \textbf{9}, 162--167.

\bibitem[{Dalling et~al.(2022)Dalling, John, Harms, Stallard and
  Yavitt}]{BCIsoil}
Dalling, J., John, R., Harms, K., Stallard, R. and Yavitt, J. (2022) Soil maps
  of {B}arro {C}olorado {I}sland 50 ha plot.
\newblock
  \urlprefix\url{http://ctfs.si.edu/webatlas/datasets/bci/soilmaps/BCIsoil.html}.
\newblock Accessed: 9 September 2022.

\bibitem[{Davison and Hinkley(1997)}]{DavisonHinkley1997}
Davison, A.~C. and Hinkley, D.~V. (1997) \textit{Bootstrap Methods and their
  Application}.
\newblock Cambridge University Press.

\bibitem[{Diggle(2010)}]{Diggle2010a}
Diggle, P.~J. (2010) Nonparametric methods.
\newblock In \textit{Handbook of Spatial Statistics} (eds. A.~E. Gelfand, P.~J.
  Diggle, M.~Fuentes and P.~Guttorp), Handbooks of Modern Statistical Methods.
  {CRC} Press, 1st edn.

\bibitem[{Duong(2007)}]{ks_package}
Duong, T. (2007) ks: Kernel density estimation and kernel discriminant analysis
  for multivariate data in {R}.
\newblock \textit{J. Stat. Softw.}, \textbf{21}, 1–16.

\bibitem[{Dvořák et~al.(2022)Dvořák, Mrkvička, Mateu and
  González}]{Dvorak2022}
Dvořák, J., Mrkvička, T., Mateu, J. and González, J.~A. (2022)
  Nonparametric testing of the dependence structure among
  points–marks–covariates in spatial point patterns.
\newblock \textit{Int. Stat. Rev.}
\newblock To appear.

\bibitem[{Fortin and Payette(2002)}]{FortinPayette2002}
Fortin, M.-J. and Payette, S. (2002) How to test the significance of the
  relation between spatially autocorrelated data at the landscape scale: A case
  study using fire and forest maps.
\newblock \textit{Ecoscience}, \textbf{9}, 213--218.

\bibitem[{Kutoyants(1998)}]{Kutoyants1998}
Kutoyants, Y. (1998) \textit{Statistical Inference for Spatial Poisson
  Processes}.
\newblock No. 134 in Lecture Notes in Statistics. New York: Springer.

\bibitem[{Lotwick and Silverman(1982)}]{Lotwick1982}
Lotwick, H.~W. and Silverman, B.~W. (1982) Methods for analysing spatial
  processes of several types of points.
\newblock \textit{J. R. Stat. Soc. Ser. B Stat. Methodol.}, \textbf{44},
  406--413.

\bibitem[{Mrkvička et~al.(2021)Mrkvička, Dvořák, González and
  Mateu}]{MrkvickaEtAl2020}
Mrkvička, T., Dvořák, J., González, J.~A. and Mateu, J. (2021) Revisiting
  the random shift approach for testing in spatial statistics.
\newblock \textit{Spat. Stat.}, 100430.

\bibitem[{Myllym{\"a}ki et~al.(2017)Myllym{\"a}ki, Mrkvi\v{c}ka, Grabarnik,
  Seijo and Hahn}]{MyllymakiEtal2017}
Myllym{\"a}ki, M., Mrkvi\v{c}ka, T., Grabarnik, P., Seijo, H. and Hahn, U.
  (2017) Global envelope tests for spatial processes.
\newblock \textit{J. R. Stat. Soc. Ser. B Stat. Methodol.}, \textbf{79},
  381--404.

\bibitem[{Nelsen(2006)}]{Nelsen2006}
Nelsen, R. (2006) \textit{An Introduction to Copulas}.
\newblock Springer New York, 2 edn.

\bibitem[{Schoenberg(2005)}]{Schoenberg2005}
Schoenberg, F.~P. (2005) Consistent parametric estimation of the intensity of a
  spatial–temporal point process.
\newblock \textit{J. Stat. Plan. Inference}, \textbf{128}, 79--93.

\bibitem[{Upton and Fingleton(1985)}]{UptonFingleton1985}
Upton, G. J.~G. and Fingleton, B. (1985) \textit{Spatial Data Analysis by
  Examples, Vol. I. Point Pattern and Quantitative Data}.
\newblock John Wiley \& Sons, New York.

\bibitem[{van~der Vaart(1998)}]{van1998asymptotic}
van~der Vaart, A. (1998) \textit{Asymptotic Statistics}.
\newblock Cambridge series on statistical and probabilistic mathematics.
  Cambridge University Press.

\bibitem[{Waage\-petersen and Guan(2009)}]{WaagepetersenGuan2009}
Waage\-petersen, R. and Guan, Y. (2009) Two-step estimation for inhomogeneous
  spatial point processes.
\newblock \textit{J. R. Stat. Soc. Ser. B Stat. Methodol.}, \textbf{71},
  685--702.

\bibitem[{Waagepetersen(2008)}]{Waagepetersen2008}
Waagepetersen, R. (2008) Estimating functions for inhomogeneous spatial point
  processes with incomplete covariate data.
\newblock \textit{Biometrika}, \textbf{95}, 351--363.

\end{thebibliography}
\bibliographystyle{rss}

\appendix

\section{Correction factors for random shifts with variance correction}\label{appendix:correction_factors}

For using the variance correction for random shift tests described in Section~\ref{subsubsec:variance_correction} the variance of the test statistic under the null hypothesis must be given as a function of the size of the observation window, at least asymptotically.

\subsection{Partial correlation coefficient $\hat \tau_p$}\label{app:A1}

Concerning the partial correlation coefficient $\hat\tau_p$, assume now that the set of sampling points $\{y_1, \ldots, y_n\}$ is a realization of a stationary point process with intensity $\lambda_0$ in $W$. Then, the variance of $\hat\tau_p$ is of order $1/(\lambda_0 |W|)$ under reasonable assumptions, see Theorem~2 of \citet{Dvorak2022}. The true intensity $\lambda_0$ of the process of sampling points is estimated as $n/|W|$, implying that a reasonable estimate of the order of variance of $\hat\tau_p$ is $1/n$. This aligns with the known result that the asymptotic order of variance of the sample Kendall's correlation coefficient computed from an i.i.d. sample of size $n$ is $1/n$ \citep[pp. 164--165]{van1998asymptotic}.

We illustrate in a short simulation experiment the practical applicability of the variance correction approach with this correction factor. Based on the arguments above, the function $f(|W|) = |W| \mathrm{var} (\hat\tau_p)$ should be approximately constant, at least for higher values of $|W|$. We consider a sequence of square observation windows $W_a = [0,a]^2$ with $a=0.5, 1, 1.5, \ldots, 4$. Let $Z_1, Z_2, Z_3$ be centered, unit variance Gaussian random fields with the exponential correlation function with the scaling parameter $\phi = 0.05, 0.10, 0.15$. Let $X$ be the log-Gaussian Cox process driven by the random intensity function $\Lambda(u) = \exp \{ 4.0 + Z_1(u) + Z_2(u) \}, u \in W_a,$ for a given value~$a$. Let the observed covariates be $C_1(u) = Z_1(u), C_2(u) = Z_1(u) + Z_3(u)$. The covariate $C_1$ describes the inhomogeneity in the point process $X$ while the unobserved random field $Z_2$ governs the interactions in $X$. Note that this model corresponds to the model $L_1^*$ considered in Section~\ref{subsec:sim_size_correlated} with the choice $b=1$.

If the influence of covariate $C_1$ is successfully removed from $X$, the residual field used to compute $\hat\tau_p$ is independent of covariate $C_2$. Therefore, this experiment indeed studies the behaviour of $\mathrm{var}(\hat\tau_p)$ under the null hypothesis.

For each size of the observation window $a$ and each scale of the random fields $\phi$ we generate 5\,000 independent realizations of $(Z_1, Z_2, Z_3, X, Y)$, where $Y$ is the process of sampling points, independent of the random fields $Z_i$ and the point process $X$. The process $Y$ is the homogeneous Poisson process on $W_a$ with intensity 100. The average number of sampling points ranges from 25 for $a=0.5$ to 1\,600 for $a=4$. The adaptive choice of bandwidth from Section~\ref{subsubsec:adaptive_bw} is used to compute $\hat\tau_p$. The sample variance of $\hat\tau_p$ is then determined from the independent realizations.

The sample variance multiplied by $|W_a|$ is given in Figure~\ref{fig:Rozptyly_LGCP_exp} (left) for the parametric smoothed residual field $s$ and Figure~\ref{fig:Rozptyly_LGCP_exp} (right) for the nonparametric smoothed residual field $\tilde{s}$. We observe that the plotted curves are nearly constant, and the correction factors $1/|W|$ are justified.

\begin{figure}[t]
    \centering
    \includegraphics[width=0.49\textwidth]{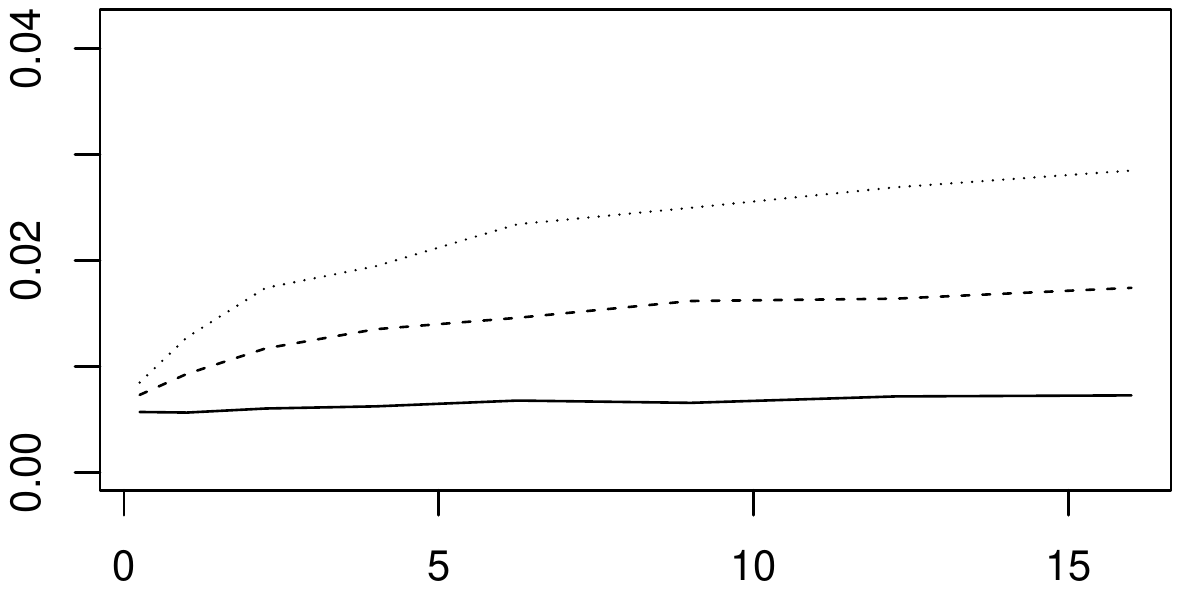}
    \includegraphics[width=0.49\textwidth]{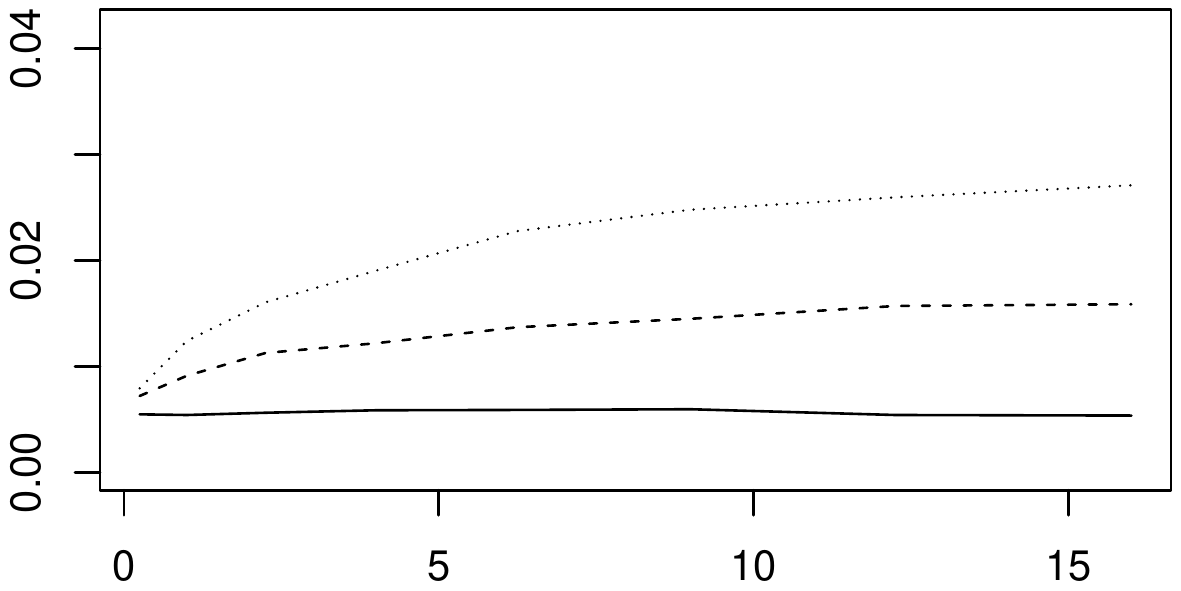}
    \caption{Sample variance of the test statistic $\hat\tau_p$, multiplied by the area of the observation window (vertical axis) plotted against the area of the observation window (horizontal axis). The scale parameter for the exponential correlation function of the random fields $Z_i$ is chosen to be 0.05 (solid line), 0.10 (dashed line) and 0.15 (dotted line), respectively. Left: for parametric residuals; right: for nonparametric residuals.}
    \label{fig:Rozptyly_LGCP_exp}
\end{figure}

\subsection{Covariate-weighted residual measure $CWR$}\label{sec:A2}

In the following proposition we show that for a Poisson process, in a simplified setting where the true intensity function $\lambda(u)$ is used to compute $CWR$, the variance of $CWR$ is a multiple of $\int_W \lambda(u) \, \mathrm{d}u$.

\begin{prop}
Let $X$ be a Poisson process on $W$ with the intensity function $\lambda(u), u \in W,$ and let $C(u), u \in W,$ be a stationary random field with $\mathbb{E} C(u)^2 = K < \infty$. Denote by $S$ the analogue of the covariate-weighted residual measure of $W$ from \eqref{eq:CWR}:
\begin{align*}
    S = \sum_{x \in X \cap W} C(x) - \int_W C(u) \lambda(u) \, \mathrm{d}u.
\end{align*}
Then $\mathrm{var} \, S = K \int_W \lambda(u) \, \mathrm{d} u.$
\end{prop}

\begin{proof}
By conditioning on $C$ and using Campbell's theorem, we obtain that
\begin{align*}
    \mathbb{E} \sum_{x \in X \cap W} C(x) = \mathbb{E} \int_W C(u) \lambda(u) \, \mathrm{d} u
\end{align*}
and hence $\mathbb{E} S = 0.$ We proceed by expressing
\begin{align*}
    \mathrm{var} \, S = \mathbb{E} S^2 = & \mathbb{E} \left( \sum_{x \in X \cap W} C(x) \right)^2 - 2 \mathbb{E} \left( \sum_{x \in X \cap W} C(x) \right) \left( \int_W C(u) \lambda(u) \, \mathrm{d} u \right) \\ & + \mathbb{E} \left( \int_W C(u) \lambda(u) \, \mathrm{d} u \right)^2
\end{align*}
and compute each term separately. First, using conditioning, Campbell's theorem, and the fact that for the Poisson process the second-order product density has the form $\lambda_2(u,v) = \lambda(u) \lambda(v), u,v \in W$, we obtain
\begin{align*}
    \mathbb{E} \left( \sum_{x \in X \cap W} C(x) \right)^2 & = \mathbb{E} \sum_{x \in X \cap W} C(x)^2 + \mathbb{E} \sum_{x, y \in X \cap W}^{\neq} C(x) C(y) \\
    & = \mathbb{E} \int_W C(u)^2 \lambda(u) \, \mathrm{d} u + \mathbb{E} \int_W \int_W C(u) C(v) \lambda(u) \lambda(v) \, \mathrm{d} u \, \mathrm{d} v.
\end{align*}

Similarly,
\begin{align*}
    -2 \mathbb{E} \left( \sum_{x \in X \cap W} C(x) \right) \left( \int_W C(u) \lambda(u) \, \mathrm{d} u \right) = - 2 \mathbb{E} \int_W \int_W C(u) C(v) \lambda(u) \lambda(v) \, \mathrm{d} u \, \mathrm{d} v,
\end{align*}
and finally
\begin{align*}
    \mathbb{E} \left( \int_W C(u) \lambda(u) \, \mathrm{d} u \right)^2 = \mathbb{E} \int_W \int_W C(u) C(v) \lambda(u) \lambda(v) \, \mathrm{d} u \, \mathrm{d} v.
\end{align*}
Altogether, we obtain
\begin{align*}
    \mathrm{var} \, S = \mathbb{E} \int_W C(u)^2 \lambda(u) \, \mathrm{d} u.
\end{align*}
Using Fubini's theorem and stationarity of $X$, we get the desired expression $\mathrm{var} \, S = K \int_W \lambda(u) \, \mathrm{d} u.$ \hfill \qedsymbol
\end{proof}

The variance of $S$ is proportional to $\int_W \lambda(u) \, \mathrm{d} u$ which is the expected number of points in $W$. In practical situations, this quantity is not known and can be estimated by the observed number of points $n(X \cap W)$.

If the intensity function is bounded from above and from below by finite positive constants, $\int_W \lambda(u) \, \mathrm{d} u$ is of order $|W|$ for large observation windows. We take advantage of this in the following simulation experiment where we determine the sample variance of $CWR$ from a set of independent realizations and standardize it by $|W|$ instead of $n(X \cap W)$ which is different for individual realizations. Furthermore, following the ideas from Theorem~1 of \citet{Dvorak2022} it can be shown that the variance of the sum in \eqref{eq:CWR} is of order $|W|$ under reasonable assumptions.


We have performed the same simulation experiment as in Section~\ref{app:A1} for $CWR$. The sample variance of $CWR$ divided $|W_a|$ is given in Figure~\ref{fig:Rozptyly_LGCP_exp_CWR}, indicating that the variance correction factor $|W|$ correctly captures the variability of $CWR$ across different realizations.

\begin{figure}[t]
    \centering
    \includegraphics[width=0.49\textwidth]{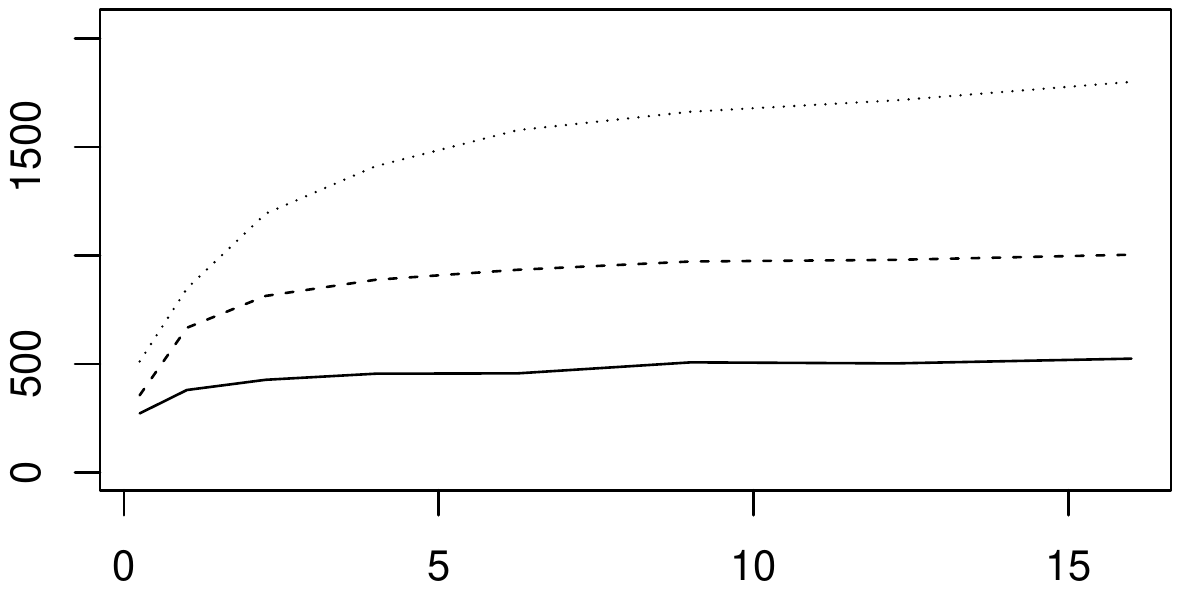}
    \includegraphics[width=0.49\textwidth]{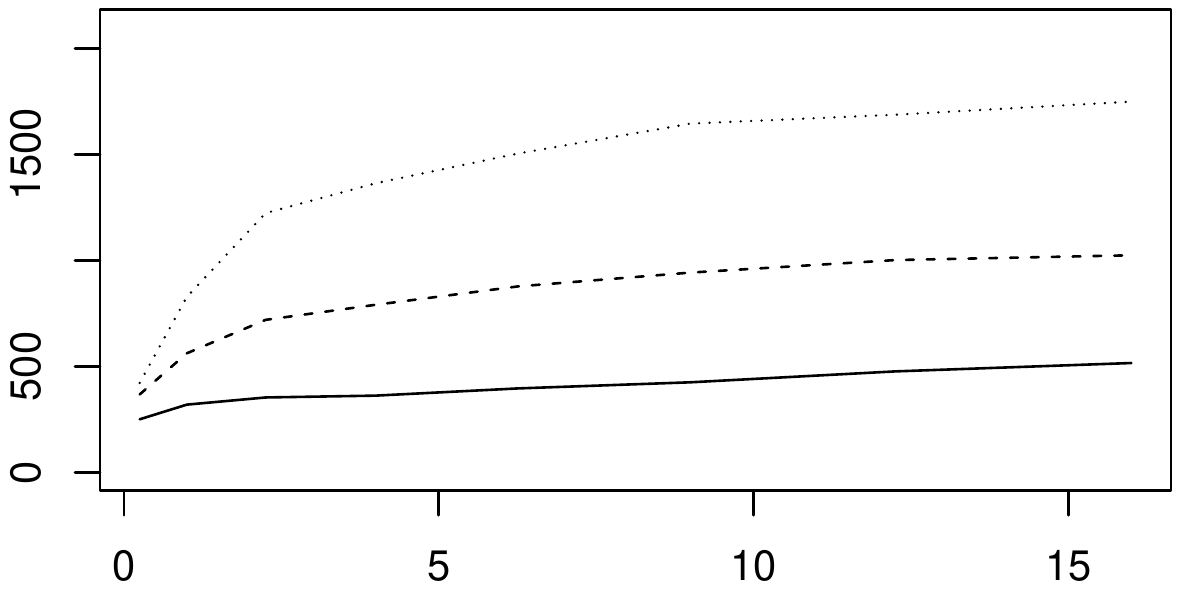}
    \caption{Sample variance of the test statistic $CWR$, divided by the area of the observation window (vertical axis) plotted against the area of the observation window (horizontal axis). The scale parameter for the exponential correlation function of the random fields $Z_i$ is chosen to be 0.05 (solid line), 0.10 (dashed line) and 0.15 (dotted line), respectively. Left: for parametric residuals; right: for nonparametric residuals.}
    \label{fig:Rozptyly_LGCP_exp_CWR}
\end{figure}

\end{document}